\documentclass{article}
\usepackage{graphicx} 
\usepackage{url}
\usepackage{amsmath,amssymb,amsthm}
\usepackage[dvipsnames]{xcolor}

\usepackage[hypertexnames=false,colorlinks=true,breaklinks=true,bookmarks=true,urlcolor=blue,citecolor=blue,linkcolor=blue,bookmarksopen=false,draft=false]{hyperref}

\makeatletter
\renewcommand*{\top}{%
  {\mathpalette\@transpose{}}%
}
\newcommand*{\@transpose}[2]{%
  \raisebox{\depth}{$\m@th#1\scriptscriptstyle\mathsf{T}$}%
}
\makeatother

\newcommand{\red}[1]{}
\newcommand{\cyan}[1]{{\color{black}{#1}}}

\DeclareMathOperator{\rank}{rank}

\DeclareMathOperator{\diag}{diag}
\DeclareMathOperator{\Diag}{Diag}

\DeclareMathOperator{\Trace}{tr}
\DeclareMathOperator{\ldet}{ldet}
\DeclareMathOperator{\dom}{dom}

\newtheorem{theorem}{Theorem}

\newtheorem{lemma}[theorem]{Lemma}

\title{Recent Advances in Maximum-Entropy Sampling}
\author{Marcia Fampa\footnote{Federal University of Rio de Janeiro}, Jon Lee\footnote{University of Michigan}}
\date{July 2, 2025; revised \today}

\begin{document}

\maketitle

\begin{abstract}
In 2022, we published 
\red{a book,}\cyan{the book} 
\emph{Maximum-Entropy Sampling: Algorithms and Application (Springer)}. Since then,
there have been several notable advancements on this topic.
In this manuscript, we survey some recent highlights.
\end{abstract}


\section*{Introduction}

Let $C$ be a symmetric positive semidefinite matrix with rows/columns
indexed from $N_n:=\{1,2,\ldots,n\}$, with $n >1$.
For $0< s < n$,
we define the \emph{maximum-entropy sampling problem}
\begin{equation}\tag{MESP}\label{MESP}
\begin{array}{ll}
z(C,s):=&\max \left\{\ldet \left(C[S(x),S(x)] \right)~:~ \mathbf{e}^\top x =s,~ x\in\{0,1\}^n\right\},
\end{array}
\end{equation}
where $S(x)$ denotes the support of $x\in\{0,1\}^n$, 
 $C[S,S]$ denotes the principal submatrix indexed by 
 $S$\cyan{, and $\ldet$ denotes the natural logarithm of the determinant}.
 For feasibility, we assume that $\rank(C)\geq s$. 
In the Gaussian case,  $\ldet (C[S,S])$ is
 proportional to the ``differential entropy'' (see \cite{Shannon}) of a vector of random variables 
 having covariance matrix $C[S,S]$. So \ref{MESP} seeks to
 find the ``most informative'' $s$-subvector from an $n$-vector following a joint 
 Gaussian distribution (see \cite{SW}). \ref{MESP} finds application in many areas, for example environmental monitoring.
 Particularly relevant for applications, we
 sometimes also consider CMESP, the \emph{constrained maximum-entropy sampling problem}, which has the additional constraints 
 $Ax\le b$.
 \ref{MESP} is NP-hard, and exact solution 
 of moderate-sized instances is approached by branch-and-bound (B\&B). 
 See \cite{FL2022} for a comprehensive treatment. 

The study of \ref{MESP} and CMESP by researchers in the mathematical-program\-ming community began with \cite{KLQ}, and remains quite active, even since the publication of \cite{FL2022}.
In what follows, we summarize very recent advances on \ref{MESP}, 
 only briefly alluded to or not at all anticipated in \cite{FL2022}. 
 For the sake of brevity, we mainly discuss \ref{MESP}, while many of the results are quite relevant for {\color{blue}CMESP}. 

 \medskip


 \noindent {\bf Notation and some key concepts}

 We let $\mathbb{S}^n_+$ (resp., $\mathbb{S}^n_{++}$) denote the set of
 positive semidefinite (resp., definite) symmetric matrices of order $n$. We let $\Diag(x)$ denote the $n\times n$ diagonal matrix with diagonal elements given by the components of $x\in \mathbb{R}^n$, and $\diag(X)$ denote the $n$-dimensional vector with elements given by the diagonal elements of $X\in\mathbb{R}^{n\times n}$.
 For a symmetric $n\times n$ matrix $U$, let
$\lambda_1(U)\ge\lambda_2(U)\ge\ldots\ge\lambda_n(U)$ denote the non-increasing ordered  eigenvalues of $U$, so
$\lambda_l(U)$ 
denotes the $l^{\rm th}$ greatest eigenvalue of $U$.
We denote  the $i$-th standard unit vector by $\mathbf{e}_i$\thinspace.
We denote an all-ones vector by $\mathbf{e}$.
For matrices $A$ and $B$ with the compatible shapes,
$A\circ B $ is the Hadamard (i.e., element-wise) product, and 
$A\bullet B:=\Trace(A^\top B)$ is the matrix dot-product.

Developing B\&B algorithms for \ref{MESP}, a maximization problem, is particularly interesting because there are several subtle upper-bounding methods. Next, we summarize a few key ones, all based on solving convex optimization problems.
But before getting to that, we briefly indicate two important concepts that can be applied to several bounding methods.

When $C$ is invertible (a common situation for many practical instances), it is easy to check that
$
z(C,s)=z(C^{-1},n-s) + \ldet C.
$
\noindent  
So we have a notion of a
\emph{complementary} \ref{MESP} problem 
\begin{align*}\tag{MESP-comp}\label{MESP-comp} 
&  \max \left\{\textstyle \vphantom{\sum_{j\in S}} \ldet C^{-1}[S(x),S(x)] ~:~ \mathbf{e}^\top x =n-s,~ x\in\{0,1\}^n\right\},
\end{align*}
 and \emph{complementary} bounds 
 (i.e., bounds for the
 complementary problem plus $\ldet C$ 
 immediately give us bounds on
 $z(C,s)$.
 Some upper bounds on $z(C,s)$ also shift by  $\ldet C$  under complementing,
 in which case there is no additional value in computing the
 complementary bound.
 
 It is also easy to check that
$
z(C,s)=z(\gamma C,s) - s \log \gamma,
$
\noindent where the \emph{scale factor} $\gamma>0$.
So we have a notion of a
\emph{scaled} \ref{MESP} problem defined by the data
 $\gamma C$, $s$, and \emph{scaled} bounds 
 (i.e., bounds for the
 scaled problem minus $s \log \gamma$) 
 immediately give us bounds on
 $z(C,s)$.
 Some upper bounds on $z(C,s)$ also shift by  $-s \log \gamma$  under scaling,
 in which case there is no additional value in computing the
scaled bound.  But otherwise, it is useful to compute a good or even optimal scale factor, and the difficulty in doing this depends on the bounding method.

For $\gamma>0$,  
the \emph{(scaled) linx  bound} for \ref{MESP}, introduced in
\cite{Kurt_linx},
is the optimal value of the convex optimization problem
\begin{align}\tag{linx}\label{linx}
    \textstyle \frac{1}{2}  
    \max \left\{
	\ldet \strut\!\!\right.&
    \left(\gamma C\Diag(x)C+\Diag(\mathbf{e}-x)\right)-s\log \gamma
	~:~ \\
    & \left. \mathbf{e}^\top x=s, ~ x\in[0,1]^n \right\}.\nonumber
\end{align}
We note that the linx bound is invariant under complementation (see \cite{Kurt_linx}).

For $\gamma>0$,  the
\emph{(scaled) BQP bound} as the optimal value of 
\begin{align}\tag{BQP}\label{BQP}
\textstyle
\max \{\ldet&\left(\gamma C\circ X + \Diag(\mathbf{e}-x) \right) - s\log(\gamma) \, :\\& \, \mathbf{e}^\top x\!=\!s,\,X\mathbf{e}\!=\!sx,\,x\!=\!\Diag(X),\, X\!\succeq\!xx^\top\}.\nonumber
\end{align}
The constraint $X\!\succeq\!xx^\top$ is the well-known convex relaxation of the nonconvex defining equation $X:=xx^\top$.
The \ref{BQP} bound was first developed for \ref{MESP} by \cite{Anstreicher_BQP_entropy}.

Now suppose that the rank of $C$ is $r\geq s$.
We factorize $C=FF^\top$,
with $F\in \mathbb{R}^{n\times k}$, for some $k$ satisfying $r\le k \le n$. Next, we define
\[
f(\Theta,\nu,\tau):= -\sum_{\ell=k-s+1}^k \log\left(\lambda_{\ell} \left(\Theta\right)\right)
+ \nu^\top \mathbf{e} +\tau s - s,
\]
and the \emph{factorization bound}, introduced in \cite{Nikolov}, is the optimal value of the 
convex optimization problem 
\[
\begin{array}{ll}
    & 	\min~ f(\Theta,\nu, \tau)\\
     & \mbox{subject to:}\\
     &\quad \diag(F \Theta F^\top) + \upsilon - \nu   - \tau\mathbf{e}=0,\\
&\quad \Theta\succ 0, ~\upsilon\geq 0, ~\nu\geq 0,
\end{array}
\tag{DFact}\label{DFact}
\]
The careful reader will notice that we do not have a scale factor $\gamma$ for the factorization bound. But this is because it is invariant under scaling.
Additionally, the factorization bound does not depend on which factorization of $C$ is chosen;
(see \cite{ChenFampaLee_Fact} for details).

It is important to note that in practice, the factorization bound is
\emph{not} calculated by directly solving \ref{DFact}.
For practical efficiency, we work in only $n$ variables with its dual, as follows (see \cite{ChenFampaLee_Fact}, for details). 

\begin{lemma}[\protect{\cite[Lemma 13]{Nikolov}}]\label{Ni13}
 Let $\lambda\in\mathbb{R}_+^k$ satisfy $\lambda_1\geq \lambda_2\geq \cdots\geq \lambda_k$\,, define $\lambda_0:=+\infty$, and let $s$ be an integer satisfying
 $0<s\leq k$. Then there exists a unique integer $i$, with $0\leq i< s$, such that
 \begin{equation*}
 \lambda_{i }>\textstyle\frac{1}{s-i }\textstyle\sum_{\ell=i+1}^k \lambda_{\ell}\geq \lambda_{i+1}~.
 \end{equation*}
\end{lemma}
\noindent Suppose that  $\lambda\in\mathbb{R}^k_+$ with  
$\lambda_1\geq\lambda_2\geq\cdots\geq\lambda_k$~. Let $\hat\imath$ be the unique integer defined by Lemma \ref{Ni13}. We define
\begin{equation}\label{def:phi}
\phi_s(\lambda):=\textstyle\sum_{\ell=1}^{\hat\imath} \log\left(\lambda_\ell\right) + (s - \hat\imath)\log\left(\frac{1}{s-{\hat\imath}} \sum_{\ell=\hat\imath+1}^{k}
\lambda_\ell\right),
\end{equation}
and, for $X\in\mathbb{S}_{+}^k$~, we define the \emph{$\Gamma$-function}
\begin{equation}\label{def:gamma}
\Gamma_s(X):= \phi_s(\lambda(X)).
\end{equation}
The factorization bound is, equivalently, the optimal value of
the convex optimization problem
\begin{align*}\label{prob_ddfact}\tag{DDFact}
\textstyle
&\max \left\{ \Gamma_s(F^\top \Diag(x)F) \, : \, \mathbf{e}^\top x=s,~
x\in[0,1]^n
\right\}. 
\end{align*}
In fact, \ref{prob_ddfact} is equivalent to the Lagrangian dual of \ref{DFact}.


\section{Complexity}


\subsection{Solvable cases}

\cite{ALTHANI2021127,ALTHANI2023120}
gave a dynamic-programming algorithm for
\ref{MESP} when the support graph of $C$ is
a spider with a bounded number of legs. 
A special case is when the support graph is a path, in which case $C$ is a tridiagonal matrix. 

In fact, the starting point for
handling spiders with a bounded number of legs is the case of a path.
The determinant of a symmetric tridiagonal matrix 
can be calculated in linear time, via a simple recursion. 
Let $T_1=(a_1)$, and  for $r\geq 2$, let
\[
T_r:=\left(
  \begin{array}{ccccc}
    a_1 & b_1 &  &  &  \\
    b_1 & a_2 & b_2 &  &  \\
     & b_2 & \ddots & \ddots &  \\
     &  & \ddots & \ddots & b_{r-1} \\
     &  &  & b_{r-1} & a_r \\
  \end{array}
\right).
\]
Defining  $\det T_0:=1$, we have 
 $\det T_r =: a_r \det T_{r-1} - b_{r-1}^2 \det T_{r-2}$,
for  $r\geq 2$. 

\begin{theorem}[\protect{\cite[Theorem 2]{ALTHANI2023120}}]
{\rm\ref{MESP}} is polynomially solvable when $C$ or $C^{-1}$ is tridiagonal, or when there is a symmetric
permutation of $C$ or $C^{-1}$ so that it is tridiagonal. 
\end{theorem}

\begin{proof} 
Without loss of generality, we may suppose that $C$ is tridiagonal. Let
$S$ be an ordered subset of $N_n$\,. Then we can write 
$C[S,S]$ uniquely as  
$C[S,S]=\Diag(C[S_1,S_1],C[S_2,S_2], \ldots, C[S_p,S_p])$,
with $p\geq 1$,
where each $S_i$ is a \emph{maximal ordered contiguous subset} of $S$, 
and for all $1\leq i<j\leq p$, all elements of $S_i$  are less than all elements of $S_j$\,. 
We refer to the $S_i$ as the \emph{pieces} of $S$, and in particular $S_p$ is the \emph{last piece}. 
It is easy to see
that 
\[
\det C[S,S] = \textstyle  \prod_{i=1}^p \det C[S_i,S_i] = \det C[S_p,S_p] \times \det C[S\setminus S_p, S\setminus S_p].
\]
Every $S$ has a last piece, and for an optimal $S$ to \ref{MESP}, if the last piece is $S_p=:[k,\ell]$,
then we have:
\[
\ldet C[S\setminus [k,\ell],S\setminus [k,\ell]]= z(C[N_{k-2},N_{k-2}], s-(\ell-k+1)).
\]
 So, we define
\begin{align*} 
f(k,\ell,t):=& \max \Big\{
\ldet C[S,S] ~:~ \\
&\qquad\qquad |S|=t,~ S\subset N_n\,, \mbox{ and the last piece of $S$ is } [k,\ell]
\Big\},
\end{align*}
for $1\leq \ell-k+1\leq t \leq s$.
We have that
\[
z(C,s) = \max_{k,\ell} \left\{ 
f(k,\ell,s) ~:~ 1\leq k\leq \ell\leq n,~ \ell-k+1\leq s 
\right\},
\]
where we maximize over the possible (quadratic number of) last pieces. 

Our dynamic-programming recursion is then 
\begin{align*}
f(k,\ell,t) = \ldet C[[k,\ell],[k,\ell]] + &
\max_{i, j} \left\{ 
f(i,j,t-(\ell-k+1)) ~:~ \right. \\
& \left. 1\leq i\ \leq j\leq k-2,~ j-i+1 \leq t-(\ell-k+1)
\right\}.
\end{align*}
To initialize, we calculate
$
f(k,\ell,\ell-k+1)=\ldet C[[k,\ell],[k,\ell]],
$
for $1\leq k \leq \ell \leq n$, $\ell-k+1\leq s$.
We can carry out the initialization in $\mathcal{O}(n^2)$ operations, using the tridiagonal-determinant formula.
Using now the recursion, for $t=1,2,\ldots,s$,
we calculate 
$f(k,\ell,t)$ for all $1\leq k \leq \ell \leq n$ such that $\ell-k+1 <t$.
We can see that this gives an $\mathcal{O}(n^5)$ algorithm for \ref{MESP}, when 
$C$ is tridiagonal. 
\end{proof}

We are interested in the case where the support graph of $C$ is a ``spider'' having,
without loss of generality, $r\geq 3$ legs
on an n-vertex set: for convenience,
we let the vertex set be $N_n$\,, and we let vertex 1 be the \emph{body} of the spider; the non-body vertex set $V_i$ 
of \emph{leg} $i$, is a non-empty 
contiguously numbered subset of $N_n \setminus\{1\}$,
such that distinct $V_i$ do not intersect, and the union of all $V_i$ is $N_n\setminus \{1\}$; 
 we number the legs $i$ in such a way that: (i) the minimum element of $V_1$ is 2, and (ii)
the minimum element of $V_{i+1}$ is 
one plus the maximum element of $V_i$\,, for $i\in[1,r-1]$.

Consider how a \ref{MESP} solution $S$ intersects with the vertices of the spider.
The solution $S$ has pieces. Note how at most one piece
contains  the body, and every other piece is a contiguous set
of vertices of a leg. The number of distinct possible pieces 
containing the body is $\mathcal{O}(n^r)$. And the number
of other pieces is $\mathcal{O}(n^2)$. Overall, we have 
$\mathcal{O}(n^r)$ pieces. In any solution, we can order the pieces
by the minimum vertex in each piece. Based on this, we have a 
well-defined last piece. From this, we can devise an efficient dynamic-programming algorithm, when we consider $r$ to be constant, and we have the following result. 

\begin{theorem}[\protect{\cite[Theorem 3]{ALTHANI2023120}}]
{\rm\ref{MESP}} is polynomially solvable when $G(C)$ or $G(C^{-1})$ is a  spider with a bounded number of legs.
\end{theorem}


\subsection{Hardness}

Related but in contrast to spiders having a bounded number of legs, we have stars, which have the maximum number of legs for a spider, but all of which are short. 
In such a case, $C$ is known as an ``arrowhead matrix'' (for a symmetric row/column permutation placing the body first or last).
\ref{MESP} was already established to be NP-hard by \cite{KLQ}, and W[1]-hard 
(a notion in parameterized complexity theory) with respect to $s$ by \cite{KOUTIS20068}. 
\cite{Ohsaka} recently established these same 
conclusions even when the support graph of 
$C$ is a star. Among other things, 
\cite{Ohsaka} also proved 
W[1]-hardness with respect to the rank of $C$. 
\cyan{Finally, \cite{MESP2DOPT}
recently proved that \ref{MESP} is NP-hard even when the covariance matrix 
$C$ is a rank-deficient matrix with all positive eigenvalues equal.}


\section{Computing}


\subsection{General purpose solvers}

For the factorization bound, \cite{ChenFampaLee_Fact}
 demonstrated that it is
invariant under scaling, and it is also independent of the particular factorization chosen.
They gave an iterative variable-fixing methodology
for use \red{heuristically and} within B\&B. 
They also demonstrated that the
factorization bound can be calculated using
a general-purpose nonlinear-programming solver (Knitro, for their experiments).
Finally, they demonstrated that the known ``mixing'' technique 
(see \cite{Mixing}) can be successfully used to 
combine the factorization bound for \ref{MESP} 
with the factorization bound for \ref{MESP-comp}, and also with
the linx bound for \ref{MESP}.


\subsection{Alternating Directions Methods of Multipliers}

\cite{ponte2025admm01doptmesp} developed successful ADMM algorithms
for rapid calculation of the linx bound, the factorization bound, and the BQP bound. In the next three subsections, we summarize those algorithms. 

\subsubsection{ADMM for the linx bound}

We rewrite 
\ref{linx}	as 
\begin{alignat*}{2}\label{prob:admmlinx1}
& \textstyle \frac{1}{2}\min ~ && -(\ldet(Z) -s\log(\gamma))\\
& \mbox{\!\!\quad s.t.} \quad &&-\left(\gamma C \Diag(x) C +  \Diag(\mathbf{e}-x)\right)  + Z = 0,\nonumber\\
&&& \mathbf{e}^\top x = s,\nonumber\\
&&& 
x\in[0,1]^n,~ Z\in\mathbb{S}^n,
\nonumber
\end{alignat*}
and then the associated augmented Lagrangian function  is 
\begin{align*}
\mathcal{L}_\rho(&x,Z,\Psi,\delta):=-\ldet(Z) \!+\! \frac{\rho}{2}\!\left\| -\gamma C \Diag(x) C- \Diag(\mathbf{e}-x)  \!+\! Z \!+\! \Psi    \right\|^2_F \\
& \qquad\qquad \qquad\qquad + \frac{\rho}{2}\!\left(-\mathbf{e}^\top x \!+\! s + \delta  \right)^2 - \frac{\rho}{2}\!\left\|\Psi\right\|_F^2 -  \frac{\rho}{2}\delta^2 +s\log(\gamma),
\end{align*}
where $\rho >0$ is the penalty parameter, and $\Psi \in \mathbb{S}^{n}$, $\delta \in \mathbb{R}$  are the scaled  Lagrangian multipliers. The associated ADMM
algorithm iteratively and successively updates
$x$, $Z$, $\Psi$ and $\delta$.
In this case, (i) the update of $x$ is a
bounded-variable least-squares problem, (ii) the update of
$Z$ has a nice closed form, and (iii) the updates of the dual variables ($\Psi$ and $\delta$) are via simple formulae (as usual).
Convergence of this ADMM algorithm is assured by general convergence theory for ADMM applied to convex problems.


\subsubsection{ADMM for the factorization bound}

We rewrite
 \ref{prob_ddfact} as  
\begin{alignat*}{1} 
& \min \left\{ -\Gamma_s(Z) ~:~
-F^\top\Diag(x)F + Z = 0,~
\mathbf{e}^\top x = s,~
x\in[0,1]^n,  Z\in\mathbb{S}^n\right\},
\end{alignat*}
and then the associated augmented Lagrangian function is 
\begin{align*}
&\mathcal{L}_\rho(x,Z,\Psi,\delta)\!:=
\!-\Gamma_s(Z) \!+\! \frac{\rho}{2}\!\left\|-F^\top\Diag(x)F \!+\! Z \!+\! \Psi  \right\|^2_F \!+\! \frac{\rho}{2}\!\left(-\mathbf{e}^\top x \!+\! s + \delta  \right)^2\\
& \qquad\qquad \qquad\qquad - \frac{\rho}{2}\!\left\|\Psi\right\|_F^2 
- \frac{\rho}{2}\delta^2\,.
\end{align*}
The ADMM updates are as described for the ADMM for linx, except the
update of $Z$, which still has a nice closed form (under some mild technical conditions), and has quite a complicated 
derivation (see \cite{ponte2025admm01doptmesp}).
Again, convergence is assured by convexity. 


\subsubsection{ADMM for the BQP bound}

We rewrite \ref{BQP} as 
\begin{equation*}
\begin{array}{rl}
\min\quad & -\ldet(Z) + s\log(\gamma)\nonumber\\
 \mbox{\!\!\quad s.t.} \quad & -( \tilde C\circ W + I_{n+1}) + Z =0,\nonumber\\
 &  W - E = 0,\label{prob:bqp}\\
 & g_\ell  - G_\ell  \bullet  W = 0, \quad \ell  = 1,\dots,2n + 2,\nonumber\\
& W,Z \in \mathbb{S}^{n+1}, ~ E \in \mathbb{S}^{n+1}_+, \nonumber
\end{array}
\end{equation*}
where $\tilde C := \begin{bmatrix}
    0 ~&~ \mathbf{0}^\top\\
    \mathbf{0} ~&~ \gamma C-I_n
\end{bmatrix}\in \mathbb{S}^{n+1}$, $W := \begin{bmatrix}
    1 ~&~ x^\top\\
    x ~&~ {X}
\end{bmatrix}\in \mathbb{S}^{n+1}$, and $g_\ell  - G_\ell  \bullet  W = 0$, with  $G_\ell \in \mathbb{S}^{n+1}$ and $g_\ell \in\,\mathbb{R}$, includes the constraints    $\Diag(X)\!=\!x$ ($\ell=1,\ldots,n$), $X\mathbf{e}\!=\! sx$ ($\ell=n+1,\ldots,2n$), $\mathbf{e}^\top x\!=\!s$ ($\ell=2n+1$), $W_{11}=1$ ($\ell=2n+2$).
Then, the   associated augmented Lagrangian function is 

\begin{align*}
\mathcal{L}_\rho(&W,E,Z,\Psi,\Phi,\omega):=\!-\ldet(Z) \!+\! \frac{\rho}{2}\!\left\|Z \!-\!\tilde{C}\circ {W} \!- \!I_{n+1}  \!+\! \Psi  \right\|^2_F  \!+\! \frac{\rho}{2}\!\left\|W \!-\!{E}  \!+\! \Phi  \right\|^2_F  \\
&\qquad\quad \!+\! \sum_{\ell =1}^{2n+2 }\!\frac{\rho}{2}\!\left(g_\ell  \!-\! G_\ell \bullet  W  \!+\! \omega_\ell   \right)^2 - \frac{\rho}{2}\!\left\|\Psi\right\|_F^2 - \frac{\rho}{2}\!\left\|\Phi\right\|_F^2 -  \frac{\rho}{2}\|\omega\|_2^2 +s\log(\gamma),
\end{align*}
Because we have three primal variable 
$W \in \mathbb{S}^{n+1}$, $E\in \mathbb{S}^{n+1}_+$ 
and $Z\in \mathbb{S}^{n+1}$,
we develop a 3-block ADMM algorithm.
%
For this, 
\red{we do not have a convergence
 guarantee}\cyan{we cannot directly apply standard results 
(see \cite{3block}, and the references therein)
to guarantee  convergence,} but \cite{ponte2025admm01doptmesp} documented practical convergence. 
Here,  (i) the update of $W$ is accomplished by solving an ordinary least-squares problem, (ii) the update of $E$ has a closed form\red{after 
solving a relatively simple semidefinite optimization problem}, and (iii) 
the update of $Z$ is very similar to the $Z$ update for linx.


\section{Bound-improvement techniques}


\subsection{Masking}

Empirically, one of the best upper bounds for \ref{MESP} is the
\ref{linx} bound. A known general technique that can potentially
improve a bound is \emph{masking}; see \cite{AnstreicherLee_Masked,BurerLee}
and its precursors \cite{HLW,LeeWilliamsILP}.
Masking means applying the
bounding method to $C\!\circ\! M$, where $M$ is any correlation matrix. 
In a precise sense, information is never gained (for any $S\subset N_n$) by masking, but an upper bound for \ref{MESP} may improve. 
\cite{maskinglinx} established that 
the (scaled) \ref{linx} bound can be improved via masking by an amount that is at least  linear in $n$,
even when optimal scaling parameters are employed. 
\cite{maskinglinx} also  extends 
 the result \cyan{of \cite{Mixing}} that the \ref{linx} bound is convex in the logarithm of the scaling parameter and
fully characterizes its behavior and provides an efficient means of calculating its limit as $\gamma$ goes to infinity. 


\subsection{Generalized scaling}

\cite{ChenFampaLeeGenScaling} generalized the scaling technique.
We refer to the technique as ``generalized scaling'' (g-scaling).
In this context, we refer the the original scaling idea as
``ordinary scaling'' (o-scaling), and we refer to a bound subject to ordinary scaling with $\gamma:=1$ 
as ``unscaled''. 
Generalized scaling manifests
differently, depending on the bound, guided by the goal of having the bound be convex in some monotone (coordinate-wise) function of the scaling parameter $\Upsilon\in\mathbb{R}^n_{++}$\,.


\subsubsection{Generalized scaling for the BQP bound}
We define the convex set 
\begin{align*} 
P(n,s):=\!\left\{
(x,X)\in\mathbb{R}^n\! \times \! \mathbb{S}^n \,:\,
 X \! - \!xx^\top\succeq 0,\, \diag(X)=x,\,
\mathbf{e}^\top x=s,\, X\mathbf{e}=sx
\right\}.
\end{align*}
For $\Upsilon \in \mathbb{R}_{++}^n$ and $(x, X)\in P(n,s)$, we define 
   \begin{align*}
   f_{{\mbox{\tiny BQP}}}(x,X;\Upsilon):=&
\textstyle \ldet \left(\strut \left(\Diag(\Upsilon)C\Diag(\Upsilon)\right)\circ X  + \Diag(\mathbf{e}-x)\right)\\
&\quad - \textstyle 2\sum_{i=1}^n  x_i\log \gamma_i \,,
	\end{align*}
\noindent with domain  
\begin{align*}
\dom\left(f_{{\mbox{\tiny BQP}}};\Upsilon\right):=
  &\left\{\strut(x,X) \in\mathbb{R}^n\times \mathbb{S}^n~:~ \right.\\
 & \quad \left.\left(\strut\Diag(\Upsilon)C\Diag(\Upsilon)\right)\circ X  + \Diag(\mathbf{e}-x) \succ 0\right\}. 
\end{align*}
The \emph{g-scaled BQP bound} is defined as
\[
\begin{array}{ll}
	& z_{{\mbox{\tiny BQP}}}(\Upsilon):=\max \left\{f_{{\mbox{\tiny BQP}}}(x,X;\Upsilon) ~:~ (x,X)\in P(n,s) \right\}.\tag{BQP-g}\label{BQP-g} 
	\end{array}
 \] 
We can interpret \ref{BQP-g} as applying the
unscaled \ref{BQP} bound to the symmetrically-scaled 
matrix $\Diag(\Upsilon)C\Diag(\Upsilon)$, and then
correcting by $- 2\sum_{i=1}^n  \!x_i\log \gamma_i$~. 

\begin{theorem}[\protect{\cite[Theorem 1]{ChenFampaLeeGenScaling}}]\label{thm:bqp}
For $\Upsilon \in \mathbb{R}_{++}^n$\,,  
we have:
\begin{itemize}
\item[\ref{thm:bqp}.i.] \label{bqp.i}  $z_{{\mbox{\tiny\rm BQP}}}(\Upsilon)$ is a valid upper bound for the optimal value of {\rm\ref{MESP}}; 
\item[\ref{thm:bqp}.ii.] \label{bqp.ii} the function $f_{{\mbox{\tiny\rm BQP}}}(x,X;\Upsilon)$ is concave in $(x,X)$ on $\dom\left(f_{{\mbox{\tiny\rm BQP}}};\Upsilon\right)$ and continuously differentiable in $(x,X, \Upsilon)$ on
$\dom\left(f_{{\mbox{\tiny\rm BQP}}};\Upsilon\right) \times \mathbb{R}_{++}^n$\,;
\item[\ref{thm:bqp}.iii.] \label{bqp.iii} for fixed $(x,X)\in \dom\left(f_{{\mbox{\tiny\rm BQP}}};\Upsilon\right)$, $f_{{\mbox{\tiny\rm BQP}}}(x,X;\Upsilon)$ is convex in $\log \Upsilon$, and so $z_{{\mbox{\tiny\rm BQP}}}(\Upsilon)$ is convex in $\log \Upsilon$. 
\end{itemize}
\end{theorem}

\noindent 
The special case of Theorem \ref{thm:bqp}.\emph{i} for the case of o-scaling is due to \cite{Anstreicher_BQP_entropy}. 
The concavity in Theorem \ref{thm:bqp}.\emph{ii}
is mainly a result of \cite{Anstreicher_BQP_entropy}, with complete details supplied in \cite[Section 3.6.1]{FL2022}. Theorem \ref{thm:bqp}.\emph{iii}
generalizes a result of \cite{Mixing}, where it is
established only for o-scaling. Theorem \ref{thm:bqp}.\emph{iii} is rather important as it enables the use of quasi-Newton methods for finding the globally-optimal g-scaling vector for the \ref{BQP} bound. 

\subsubsection{Generalized scaling for the linx bound}

 For  $\Upsilon \in \mathbb{R}_{++}^n$ and  $x\in[0,1]^n$,
 we define 
\begin{align*} 
f_{{\mbox{\tiny linx}}}(x;\Upsilon) :=& \textstyle\frac{1}{2}\left(\strut\ldet \left( \Diag(\Upsilon) C\Diag(x)C \Diag(\Upsilon)+\Diag(\mathbf{e}-x) \right)\right) \\ &\quad \textstyle -\sum_{i=1}^n x_i\log \gamma_i \,,
\end{align*}
with  
\begin{align*}
 \dom\left(f_{{\mbox{\tiny linx}}};\Upsilon\right)\!:=\!  \left\{\strut x \in \mathbb{R}^n \!\,:\!\, \Diag(\Upsilon) C\Diag(x)C \Diag(\Upsilon)+\Diag(\mathbf{e}-x) \succ 0 \right\}.
\end{align*}
We then define the \emph{g-scaled linx  bound}  
\begin{equation}\tag{linx-g}\label{linx-g} 
\begin{array}{ll}
	z_{{\mbox{\tiny linx}}}(\Upsilon):=\max \left\{\strut f_{{\mbox{\tiny linx}}}(x;\Upsilon) 
	~:~ \mathbf{e}^{\top}x=s,~ 0\leq x\leq \mathbf{e}\strut \right\}.
\end{array}
\end{equation}
In contrast to \ref{BQP-g},
we cannot interpret \ref{linx-g} as applying the
unscaled linx bound to a symmetrically diagonally scaled $C$.

\begin{theorem}[\protect{\cite[Theorem 2]{ChenFampaLeeGenScaling}}]\label{thm:linx}
For  $\Upsilon \in \mathbb{R}_{++}^n$\,, we have:
\begin{itemize}
\item[\ref{thm:linx}.i.] $z_{{\mbox{\tiny\rm linx}}}(\Upsilon)$ is a valid upper bound for the optimal value of {\rm\ref{MESP}}; 
\item[\ref{thm:linx}.ii.] the function $f_{{\mbox{\tiny\rm linx}}}(x;\Upsilon)$ is concave in $x$ on $\dom\left(f_{{\mbox{\tiny\rm linx}}};\Upsilon\right)$ and continuously differentiable in $(x,\Upsilon)$ on $\dom\left(f_{{\mbox{\tiny\rm linx}}};\Upsilon\right)\times \mathbb{R}^n_{++}$\,;
\item[\ref{thm:linx}.iii.] \label{linx.iii} for fixed $x\in \dom\left(f_{{\mbox{\tiny\rm linx}}};\Upsilon\right)$, $f_{{\mbox{\tiny\rm linx}}}(x;\Upsilon)$ is convex in $\log \Upsilon$, and thus $z_{{\mbox{\tiny\rm linx}}}(\Upsilon)$ is convex in $\log \Upsilon$. 
\end{itemize}
\end{theorem}

\noindent The special case of  Theorem \ref{thm:linx}.\emph{i} for o-scaling 
was established by \cite{Kurt_linx}.
The concavity in Theorem \ref{thm:linx}.\emph{ii}
is mainly a result of \cite{Kurt_linx}, with further details 
supplied in \cite{FL2022}.
The special case of Theorem \ref{thm:linx}.\emph{iii} for o-scaling was established by \cite{Mixing}. As for the
g-scaled \ref{BQP} bound, the result is rather important as it enables the use of quasi-Newton methods for finding the globally optimal g-scaling for the \ref{linx} bound.


\subsubsection{Generalized scaling for the factorization bound}

For
$\Upsilon \in \mathbb{R}_{++}^n$ and 
$x\in [0,1]^n$, we define  
\begin{align*}
    &F_{{\mbox{\tiny DDFact}}}(x;\Upsilon):
=\textstyle \sum_{i=1}^n \gamma_i x_i F_{i\cdot}^\top F_{i\cdot}\,, \text{ and}\\
&\textstyle f_{{\mbox{\tiny DDFact}}}(x;\Upsilon):= \strut \Gamma_s(F_{{\mbox{\tiny DDFact}}}(x;\Upsilon)) -\sum_{i=1}^n x_i \log \gamma_i\,.
\end{align*}
\noindent We define the \emph{g-scaled factorization bound}
\begin{align}
	z_{{\mbox{\tiny DDFact}}}(\Upsilon):= 	\max \left\{  f_{{\mbox{\tiny DDFact}}}(x;\Upsilon ) \,:\,\mathbf{e}^\top x=s,~
 0\leq x\leq \mathbf{e}\strut \right\}.
\tag{DDFact-g}\label{DDFact-g}
\end{align}
Note that 
$\left(\!\Diag(\sqrt{\Upsilon}) F\!\right)\!\! \left(\!\Diag(\sqrt{\Upsilon}) F\!\right)^{\top\!}$\! is a factorization of $\Diag(\!\sqrt{\Upsilon})C\Diag(\!\sqrt{\Upsilon})$, 
so we can
can interpret \ref{DDFact-g} as applying the
unscaled factorization bound to the symmetrically-scaled  matrix $\Diag(\sqrt{\Upsilon}) C \Diag(\sqrt{\Upsilon})$. 


\cyan{In the following result,
we write
\begin{align*}
    &\dom\left(\Gamma_s\right):=  \left\{X : X\succeq 0, \rank(X)\ge s \right\}, \text{ and }\\
    &\dom\left(f_{{\tiny\mbox{DDFact}}}; \Upsilon\right):=  \left\{x : F_{{\tiny\mbox{DDFact}}}(x;\Upsilon)\in \dom \left(\Gamma_s\right)\right\}
\end{align*}
for the domains of $\Gamma_s(X)$ and $f_{{\tiny\mbox{DDFact}}}(x;\Upsilon)$, respectively. Moreover, we employ 
$\dom\left(f_{{\tiny\mbox{DDFact}}}; \Upsilon\right)_+$
for the intersection of $\dom\left(f_{{\tiny\mbox{DDFact}}}; \Upsilon\right)$ and $\mathbb{R}^n_+$\,.}


\begin{theorem}[see \protect{\cite[Theorem 6]{ChenFampaLeeGenScaling}} for a more detailed statement]\label{thm:fact}
For $\Upsilon \in \mathbb{R}_{++}^n$\, we have:
\begin{itemize}
\item[\ref{thm:fact}.i.] $z_{{\mbox{\tiny\rm DDFact}}} (\Upsilon)$ yields a valid upper bound for the optimal value of {\rm\ref{MESP}};
\item[\ref{thm:fact}.ii.] the function $f_{{\mbox{\tiny\rm DDFact}}}(x;\Upsilon)$  is concave in $x$ on $\dom\left(f_{{\mbox{\tiny\rm DDFact}}}; \Upsilon\right)_+$\,; 
\item[\ref{thm:fact}.iii.] the function $f_{{\mbox{\tiny\rm DDFact}}}(x;\Upsilon)$ is ``generalized differentiable'' with respect to \break $\dom\left(f_{{\mbox{\tiny\rm DDFact}}}; \Upsilon\right)_+$\,;
\item[\ref{thm:fact}.iv.] given $x\in \dom\left(f_{{\mbox{\tiny\rm DDFact}}}; \Upsilon\right)_+$\,, the function $f_{{\mbox{\tiny\rm DDFact}}}(x;\Upsilon)$ is differentiable in $\Upsilon$; 
additionally, if $x:=x^*$, an optimal solution to \ref{prob_ddfact}, then the gradient vanishes at $\Upsilon=\mathbf{e}$;
\item[\ref{thm:fact}.v.]
the function $f_{{\mbox{\tiny DDFact}}}(x;\Upsilon)$ is continuously generalized differentiable in $x$ and continuously differentiable in $\Upsilon$ on $\dom\left(f_{{\mbox{\tiny\rm DDFact}}}; \Upsilon\right)_+ \times \mathbb{R}^n_{++}$\,.
\end{itemize}
\end{theorem}


\noindent \cite{Nikolov} (also see \cite{Weijun}) established Theorem \ref{thm:fact}.\emph{i} for $\Upsilon:=\mathbf{e}$. \red{We generalize}\cyan{Theorem \ref{thm:fact}.\emph{i} generalizes} this result to the situation where $\Upsilon\in \mathbb{R}^n_{++}$\,. 
\cite{ChenFampaLee_Fact} showed that the o-scaled factorization bound for (C)\ref{MESP} is
invariant under the scale factor, so the
use of any type of scaling in the context of the factorization bound 
\cyan{was new}\red{is new}.
Theorem \ref{thm:fact}.\emph{ii} is a result of \cite{Nikolov}, with details supplied by \cite[Section 3.4.2]{FL2022}.
Theorem \ref{thm:fact}.\emph{iii} is the first differentiablity result of any type for the 
factorization bound. 
This result helps us to understand the practical success of general-purpose codes (like Knitro) for calculating the factorization bound. 
Theorem \ref{thm:fact}.\emph{iv} provides the potential for fast algorithms leveraging Newton and quasi-Newton based methods to improve the factorization bound by g-scaling. 
We are left with the
open question of whether g-scaling can improve the factorization bound for \ref{MESP} --- we do have experimental evidence that it 
\red{case}\cyan{can}
improve the factorization bound for 
CMESP\cyan{; see \cite[Section 6]{ChenFampaLeeGenScaling}}.
We can interpret the last part of Theorem \ref{thm:fact}.\emph{iv} as a partial result toward a negative answer.
Theorem \ref{thm:fact}.$v$ is a consequence of Theorems \ref{thm:fact}.$iii,iv$.


\subsection{The augmented factorization bound}

\cite{li2025augmentedfactorizationIPCO}
recently 
gave an improvement on the factorization bound, 
for the case in which $C$ is positive definite --- an important special case. 
\cyan{
Considering the function $\phi_s$ defined in \eqref{def:phi},  she defines the \emph{$\Gamma^+$-function} for $X\in\mathbb{S}^k_+$ and $0\leq \kappa\leq \lambda_n(C)$, as
\begin{equation*}
\Gamma_s^+(X;\kappa):= \phi_s(\lambda(X) + \kappa \,
\mathbb{I}_s),
\end{equation*}
where $\mathbb{I}_s\in\mathbb{R}^k$ has the first $s$ elements equal to one and the others equal to zero. 
Then, she defines the  \emph{augmented factorization bound} as the optimal value of the convex optimization problem
\begin{align*}
\max \left\{ \Gamma^+_s(G^\top \Diag(x)G;\kappa) \, : \, \mathbf{e}^\top x=s,~
x\in[0,1]^n
\right\}, 
\end{align*}
where $GG^\top:=C - \kappa I_n$\,, for $0<\kappa\leq \lambda_n(C)$,
with  $G \in \mathbb{R}^{n\times q}$, for some $q$ satisfying $\rank(G)\le q \le n$.

\cite{li2025augmentedfactorizationIPCO} establishes
that this new bound is optimized for 
$\kappa:=\lambda_n(C)$ and dominates the factorization bound
(for this special case in which $C$ is positive definite).
}

\red{
She simply applied the
factorization bound to $C-\lambda I_n$\,, for some $0<\lambda\leq \lambda_{\text{min}}(C)$, and adds $s\lambda$.
\cite{li2025augmentedfactorizationIPCO} establishes
that this new bound is optimized for 
$\lambda:=\lambda_{\text{min}}(C)$ and dominates the factorization bound
(for this special case in which $C$ is positive definite).
}


\section{Cousins of MESP}

The 0/1 D-Optimality problem can be formulated as 
\begin{align*}\label{prob01}\tag{D-Opt(0/1)}
&\max \left\{\ldet(A^\top \Diag(x) A) \, : \, \mathbf{e}^\top x=s,~   x\in\{0,1\}^n\right\},
\end{align*} 
where $A:= (v_1, v_2, \dots,v_n)^\top\in\mathbb{R}^{n\times m}$ has full column rank, and $m\leq s < n$. 
The $v_i\in\mathbb{R}^m$ are potential
``design points'' for some experiments.
Given a budget for $s$ experiments, we
wish to minimize the generalized variance 
of parameter estimates for a linear model
based on the chosen experiments. More details can be found in \cite{PonteFampaLeeMPB}, and the references therein.

The 0/1 
\red{D-optimal}\cyan{D-Optimal} 
Data Fusion problem can be formulated as 
\begin{align*}\label{datafusion}\tag{DDF(0/1)}
&\max \left\{\red{B+}\ldet(\cyan{B+} A^\top \Diag(x) A) \, : \, \mathbf{e}^\top x=s,~   x\in\{0,1\}^n\right\},
\end{align*} 
where $B\in\mathbb{S}_{++}^m$ is an existing Fisher Information Matrix (FIM), $1\leq s < n$, and $A$ is defined as for \ref{prob01}. More details can be found in \cite{li2022d} and the references therein.

The difference between \ref{datafusion} and \ref{prob01} is that in the
\red{latter}\cyan{former} 
we assume the existence of information from previously selected experiments, represented by the existing positive definite FIM (i.e., $B:=\tilde{A}^\top\tilde{A}$, where the rows of $\tilde{A}$ correspond to the  previously selected design points), in addition to the information obtained from experiments corresponding to $n$ potential design points from which $s$ new points should be selected.
\red{(note that in this case we can have $s<m$).}


\subsection{0/1 D-Opt, 0/1 Data Fusion, and 
MESP}\label{sec:equiv}

In the following, we collect recent results from  the literature that address the relationship between \ref{prob01} and \ref{datafusion},  and \ref{MESP}. We demonstrate \emph{practical} reductions, that do not increase the sizes of instances. 

\red{
Let $C\in\mathbb{S}^n_{++}$\,,  $x\in\{0,1\}^n$ with $\mathbf{e}^\top x =s$.
Consider the factorization $\frac{1}{\lambda_n(C)} C - I_n=W^\top W$, 
where $W\in\mathbb{R}^{n\times n}$. 
For example, $W$ could be the matrix square root, or it could be derived from the real Schur decomposition.
\cite{li2022d} observed that 
\begin{equation*}
    \begin{array}{ll}
    \textstyle\ldet C[S(x),S(x)]&=\ldet\left(\frac{1}{\lambda_n(C)} C[S(x),S(x)]\right) - s\log\left(\textstyle\frac{1}{\lambda_n(C)}\right)\\
    &=\ldet((I_n+W \Diag(x)W^\top) - s\log\left(\textstyle\frac{1}{\lambda_n(C)}\right),
    \end{array}
    \end{equation*}
and from that they obtain the following result.
}

\begin{theorem}[\protect{\cite[Theorem 1]{li2022d}}]
Every instance of 
{\rm\ref{MESP}} 
having $C\in\mathbb{S}^n_{++}$ 
can be recast as an instance of 
{\rm\ref{datafusion}}, where the $n$ of the two problems are identical. 
\end{theorem}

\cyan{
\begin{proof}
Consider the factorization $\frac{1}{\lambda_n(C)} C - I_n=W^\top W$, 
where $W\in\mathbb{R}^{n\times n}$. 
For example, $W$ could be the matrix square root, or it could be derived from the real Schur decomposition.
Observe that, for
$x\in\{0,1\}^n$ with $\mathbf{e}^\top x =s$, we have
\begin{equation*}
    \begin{array}{ll}
    \textstyle\ldet C[S(x),S(x)]&=\ldet\left(\frac{1}{\lambda_n(C)} C[S(x),S(x)]\right) - s\log\left(\textstyle\frac{1}{\lambda_n(C)}\right)\\
    &=\ldet((I_n+W \Diag(x)W^\top) - s\log\left(\textstyle\frac{1}{\lambda_n(C)}\right),
    \end{array}
    \end{equation*}
and the result follows.
\end{proof}
}

\cyan{
\begin{theorem}[\protect{\cite[Theorem 2]{li2022d}}]
Every instance of 
{\rm\ref{datafusion}}
can be recast as an instance of 
{\rm\ref{MESP}} 
having $C\in\mathbb{S}^n_{++}$\,, where the $n$ of the two problems are identical. 
\end{theorem}

\begin{proof}
Observe that, for
$x\in\{0,1\}^n$ with $\mathbf{e}^\top x =s$, we have
\begin{align*}
\ldet(B + A^\top \Diag(x)A )
=& \ldet(B) +\ldet \left(I_m + B^{\scriptscriptstyle -1/2}A^\top \Diag(x)A
B^{\scriptscriptstyle -1/2} \right)\\
=& \ldet(B) +\ldet \left(I_n + \Diag(x)^{\frac{1}{2}}A  B^{-1}A^\top \Diag(x)^{\frac{1}{2}}\right)\\
=& \ldet(B) + \ldet(C[S(x),S(x)]),
\end{align*}
where $C:=I_n+AB^{-1}A^\top$.
The result follows.
\end{proof}
}

Let $x\in\{0,1\}^n$ with $\mathbf{e}^\top x =s$, and  $T(x):=N_n\setminus S(x)$.
\cite[see Remark 8]{PonteFampaLeeMPB} observed that
    \begin{equation}\label{relMESPDopt1}
    \ldet(A^\top\Diag(x)A)=2\textstyle\sum_{i= 1}^m \log(\Sigma[i,i]) + \ldet((I_n-UU^\top)[T(x),T(x)]),
    \end{equation}
where $A=U\Sigma V^\top$ is the real singular value decomposition of $A$.
From this, they established the following two results.

\begin{theorem}[\protect{\cite[see Remark 8]{PonteFampaLeeMPB}}]
Every instance of 
{\rm\ref{prob01}}
can be recast as an instance of
{\rm\ref{MESP}}, where the $n$ of the two problems are identical.
\end{theorem}

\begin{proof}
 From \eqref{relMESPDopt1}, we see that  \emph{any} instance of \ref{prob01} can be reduced to an instance of \ref{MESP}, where we search for a  maximum (log-)determinant principal submatrix  of  order $n-s$, from
the input positive-semidefinite matrix $I_n-UU^\top$ of order $n$ and rank $n-m$. 
\end{proof}

\begin{theorem}[\protect{\cite[see Remark 8]{PonteFampaLeeMPB}}]
Every instance of 
{\rm\ref{MESP}} 
having all positive eigenvalues identical
can be recast as an instance of
{\rm\ref{prob01}}, where the $n$ of the two problems are identical.
\end{theorem}

\begin{proof}
The instance of \ref{prob01} seeks to select $n-s$ rows of the input matrix  $A:=U\in\mathbb{R}^{n\times m}$, where $UU^\top= I_n- \textstyle\frac{1}{\lambda_1(C)}C$ and $U^\top U=I_m$ (i.e., $UU^\top$ is the compact spectral decomposition of $I_n-\textstyle\frac{1}{\lambda_1(C)}C$, and all of its nonzero eigenvalues are $1$). 
\end{proof}

\cyan{
In the recent work \cite{MESP2DOPT},
we established that \ref{MESP}
is fully equivalent to a slightly more general version of
\ref{prob01} that subsumes
\ref{datafusion}. Specifically,
the version (which arises, for example, as B\&B subproblems
with respect to \ref{prob01})
is simply \ref{datafusion}
with the relaxed assumption that 
$B\in\mathbb{S}^n_+$\,.
Further in \cite{MESP2DOPT}, we study in detail the behavior of objective-value upper bounds, in the context of
various maps between \ref{MESP}
instances and these more general
\ref{prob01} instances.
}


\subsection{GMESP}

The \emph{generalized maximum-entropy sampling problem}, introduced by \cite{WilliamsPhD,LeeLind2020}, has a 
similar formulation with \ref{MESP}.
It is
\begin{equation}\tag{GMESP}\label{GMESP}
\begin{array}{ll}
&\max \left\{\sum_{\ell=1}^t \log (\lambda_\ell(C[S(x),S(x)])) ~:~ 
\mathbf{e}^\top x =s,~ x\in\{0,1\}^n\vphantom{\sum_{\ell=1}^t }\right\}.
\end{array}
\end{equation}
where 
\ref{GMESP} is a natural generalization of both \ref{MESP} and 
\ref{prob01} (see \cite{LeeLind2020} for details).
In the general case (i.e., not a \ref{MESP} instance and not a \ref{prob01} instance), it is
motivated by a particular selection problem in the context of principal component analysis (PCA); see \cite{GMESParxiv}.

Following the idea of the factorization bound for \ref{MESP}, \cite{SEA_proceedings,GMESParxiv} introduced the first convex-optimization based relaxation for \ref{GMESP}, studied its behavior, compared it to an earlier spectral bound, and demonstrated its use in a B\&B scheme. Empirically, the approach seems to be effective only when $s-t$ is very small. 

In what follows, we work toward presenting a new result from \cite{GMESParxiv} concerning so-called variable fixing. The result, derived for \ref{GMESP}, in the context of the generalized factorization bound, is even new for the special case of \ref{MESP}, in the context of the factorization bound. We will present only statements and proof sketches of the special case of the results for \ref{MESP}. The full proofs and in the greater generality of \ref{GMESP} can be found in \cite{GMESParxiv}. 

First, we recall the principle of \emph{variable fixing} for \ref{MESP}, in the context of the factorization bound.
\begin{theorem}[see \cite{FL2022}] \label{thm:fixFact}
 Let
\begin{itemize}
    \item\!$\mbox{LB}$ be the objective-function value of a feasible solution for {\rm\ref{MESP}},
    \item\!$(\!\hat{\Theta},\hat{\upsilon},\hat{\nu},\hat{\tau}\!)$ be a feasible solution for {\rm\ref{DFact}} with objective-function value~$\hat{\zeta}$.
\end{itemize}
Then, for every optimal solution $x^*$ 
for {\rm\ref{MESP}}, we have:
\[
\begin{array}{ll}
x_j^*=0, ~ \forall ~ j\in N_n \mbox{ such that } \hat{\zeta}-\mbox{LB} < \hat{\upsilon}_j\thinspace,\\
x_j^*=1, ~ \forall ~ j\in N_n \mbox{ such that } \hat{\zeta}-\mbox{LB} < \hat{\nu}_j\thinspace.\\
\end{array}
\]
\end{theorem}

We note that to apply the variable-fixing procedure described in Theorem \ref{thm:fixFact} in a  B\&B algorithm to solve \ref{MESP}, we need a feasible solution for \ref{DFact}. 
\cite{Weijun} showed  how to construct a feasible solution for \ref{DFact} from  a feasible solution $\hat x$ of \ref{prob_ddfact}  with finite objective value, with the goal of producing a small gap.

Considering the spectral
decomposition $F(\hat{x})=\sum_{\ell=1}^{k} \hat \lambda_\ell \hat u_\ell \hat u_\ell^\top\,,$
with $\hat \lambda_1\geq\hat \lambda_2\geq\cdots\geq \hat \lambda_{\hat r}>\hat \lambda_{\hat{r}+1}=\cdots=\hat \lambda_k=0$,  
following \cite{Nikolov}, they  set 
$\hat{\Theta}:=\sum_{\ell=1}^{k} {\hat \beta}_\ell \hat{u}_\ell \hat{u}_\ell^\top$\,,
where
\begin{equation}\label{defbetaa}
\hat{\beta}_\ell:=\left\{
\begin{array}{ll}
        \textstyle 1/\hat{\lambda}_\ell\,,
       &~1\leq \ell\leq \hat{\iota};\\
     1/\hat{\delta},&~\hat{\iota}<\ell\leq \hat{r};\\
     (1+\epsilon)/\hat{\delta},&~\hat{r}<\ell\leq k,
\end{array}\right.
\end{equation}
 for any $\epsilon>0$, where $\hat{\iota}$ is the unique integer defined  in Lemma \ref{Ni13} for $\lambda:=\hat{\lambda}$, and
$
\hat \delta:=\frac{1}{s-\hat \iota}\sum_{\ell=\hat \iota+1}^{k}\hat \lambda_\ell
$\thinspace.
We can verify that
\begin{equation} \label{primaldual}
\textstyle
- \sum_{\ell=1}^s \log(\hat{\beta}_{\ell})=  \sum_{\ell=1}^{\hat{\iota}} \log(\hat{\lambda}_{\ell}) + (s-\hat{\iota})\log(\hat{\delta})= \Gamma_t(F(\hat{x})).
\end{equation}
Then, the minimum duality gap between $\hat x$  in \ref{prob_ddfact} and feasible solutions
of \ref{DFact}  of the form $(\hat\Theta,\upsilon,\nu,\tau)$,
is the optimal value of
\begin{equation}\label{mingapproba}\tag{$G(\hat\Theta)$}
\begin{array}{ll}
\min&~
    \nu^\top \mathbf{e}   +\tau s - s\\
      \mbox{s.t.} 
&     ~ \upsilon - \nu   - \tau\mathbf{e}= - \diag(F \hat \Theta F^\top) ,\\
&~\upsilon\geq 0, ~\nu\geq 0.
\end{array}
\end{equation}
\ref{mingapproba} has a simple closed-form solution.
 To construct it,  consider the permutation $\sigma$ of the indices in $N_n$\,, such that $\diag(F\hat\Theta F^\top)_{\sigma(1)} \geq \dots \geq \diag(F\hat\Theta F ^\top)_{\sigma(n)}$\,.
An optimal solution of \ref{mingapproba} is given by (see \cite{{Weijun},FL2022})
\begin{align*}
    &\tau^* :=  
    \diag(F\hat\Theta F^\top)_{\sigma(s)}\,,\\
    &\nu^*_{\sigma(\ell)} := \begin{cases}
         \diag(F\hat\Theta F^\top)_{\sigma(\ell)} - \tau^*,\quad &\text{for } 1\leq \ell \leq s;\\
         0,&\text{otherwise},
    \end{cases}
\end{align*}
and $\upsilon^*:=\nu^*+\tau^*\mathbf{e} - \diag(F \hat \Theta F^\top)$.

\begin{lemma}[see \cite{GMESParxiv} for a more general version of this result (for GMESP)]\label{lem:optimal_xhat_diagf}
   Let $\hat{x}$ be an optimal  solution of {\rm\ref{prob_ddfact}}. Let  
     $F(\hat x)= F^\top \Diag(\hat x) F =: \sum_{\ell=1}^{k} \hat\lambda_{\ell} \hat u_{\ell} \hat u_{\ell}^{\top}$ be a spectral decomposition of $F(\hat x)$. Let
       $\hat{\Theta}:=\sum_{\ell=1}^{k} {\hat \beta}_\ell \hat{u}_\ell \hat{u}_\ell^\top$\,, where $\hat\beta$ is defined in \eqref{defbetaa}.
Then, for every $i,j\in N_n$\,, 
we have
\begin{enumerate}
    \item[\rm(a)] \label{item1} $\diag(F\hat\Theta F^\top)_{i} \geq \diag(F\hat\Theta F^\top)_{j}\,$, if $\hat{x}_i > \hat{x}_j$\,,
    \item[\rm(b)] \label{item2} $\diag(F\hat\Theta F^\top)_{i} = \diag(F\hat\Theta F^\top)_{j}\,$, if $\hat{x}_i\,, \hat{x}_j\in(0,1)$.
    \end{enumerate}
\end{lemma}

\begin{proof}(sketch)
Let $\tilde{x}$ be a feasible solution to the \ref{prob_ddfact}.
From \cite[Proposition 11]{ChenFampaLeeGenScaling}, we have that  the directional derivative of $\Gamma_s$ at $\hat{x}$ in the direction $\tilde{x}-\hat{x}$ exists, and is given by 
    \[
    (\tilde{x}-\hat{x})^\top\textstyle\frac{\partial  \Gamma_s(F(\hat x))}{\partial  x} = (\tilde{x}-\hat{x})^\top\diag(F\hat\Theta F^\top).
    \]
Then,  because  \ref{prob_ddfact} is a convex optimization problem with a concave objective function $\Gamma_s$\,, we  conclude that $\hat{x}$ is an optimal solution to \ref{prob_ddfact} if and only if  
\begin{equation}\label{derivative}(\tilde{x}-\hat{x})^\top\diag(F\hat\Theta F^\top) \leq 0,
\end{equation}
for every feasible solution $\tilde{x}$ to \ref{prob_ddfact}.

It is possible to prove both results (a) and (b) by contradiction, because assuming any of them does not hold for some pair $i,j\in N_n$, we can construct a feasible solution $\tilde{x}$ to \ref{prob_ddfact} that contradicts \eqref{derivative}; see \cite{GMESParxiv} for details. 
\end{proof}

\begin{theorem}[see \cite{GMESParxiv} for a more general version of this result (for GMESP)]\label{dualoptimallcfa}
   Let $\hat{x}$ be an optimal solution of {\rm\ref{prob_ddfact}}. Then,  $(\hat{\Theta},\upsilon^*,\nu^*,\tau^*)$  is an optimal solution  to {\rm\ref{DFact}}. 
\end{theorem}

\begin{proof}(sketch) 
Considering \eqref{primaldual}, it suffices to prove that the objective value of \ref{mingapproba} at $(\upsilon^*,\nu^*,\tau^*)$ is zero,  that is
$ {\nu^*}^\top \mathbf{e}    +\tau^* s - s\! =\! 0$. 
 So, it suffices to show~that
\begin{equation}\label{toprovea}
    \textstyle\sum_{\ell=1}^s  \diag(F \hat\Theta F^\top)_{\sigma(\ell)} 
      = s.
\end{equation}
    We can verify (see \cite{GMESParxiv} for details) that 
    $\hat{x}^\top \diag(F \hat\Theta F^\top)= F(\hat x)\bullet \hat \Theta= s$. Then, to show \eqref{toprovea}, it suffices to show that 
        \begin{equation}\label{toprove}
       \textstyle\sum_{\ell=1}^s  \diag(F \hat\Theta F^\top)_{\sigma(\ell)} 
     = \hat{x}^\top \diag(F \hat\Theta F^\top).
       \end{equation}
  If   $\hat x \in \{0,1\}^n$,  then \eqref{toprove}  follows directly from Lemma \ref{lem:optimal_xhat_diagf}, part (a),  and the ordering defined by $\sigma$.  

  Next, suppose that $\hat{x} \notin \{0,1\}^n$. Let  
       $\mathcal{I}_1:=\{i\in  N: \hat{x}_{i}=1\}$ and $\mathcal{I}_f:=\{i\in N: \hat{x}_{i}\in(0,1)\}$. Note that $\sum_{i\in\mathcal{I}_f}\hat{x}_i=s-|\mathcal{I}_1|$. 
    Let $\hat{d}:= \diag(F\hat{\Theta}F^\top)_{i}$\,, for every $i \in \mathcal{I}_f$ (this is well defined, due to  Lemma \ref{lem:optimal_xhat_diagf}, part (b)).
    Then, 
   \[
   \hat{x}^\top \diag(F \hat\Theta F^\top)
        =  \textstyle\sum_{i \in \mathcal{I}_1}  \diag(F \hat\Theta F^\top)_{i} + (s-|\mathcal{I}_1|) \hat{d}\,.
\]
 Note that  $|\mathcal{I}_f|>s-|\mathcal{I}_1|$. Furthermore,   from Lemma \ref{lem:optimal_xhat_diagf}, part (a), we see  that $\diag(F\hat{\Theta}F^\top)_i \geq \hat{d}$ for all $i \in \mathcal{I}_1$\,. Then, we also have
   \begin{align*}
          &\textstyle\sum_{\ell=1}^s  \diag(F \hat\Theta F^\top)_{\sigma(\ell)}= \textstyle\sum_{i \in \mathcal{I}_1} \diag(F \hat\Theta F^\top)_{i} + (s-|\mathcal{I}_1|)\hat{d}~.
       \end{align*}   
       The result follows.
\end{proof}
 

\section{On the horizon} 

We have recently initiated some new work on \emph{the maximum-entropy sampling clustering problem},
which is \ref{MESP} with a particular combinatorial constraint.
The constraint is defined via an undirected graph on vertex set $N_n$\,.
Instead of just looking for an $s$-subset of $N_n$ having maximum entropy, we
have to be sure that the subgraph of $G$ induced by $S$ is complete.
If $C=I_n$\,, then the problem is to find an $s$-clique of $G$, already NP-hard.
So this is a nice merge of the $s$-\emph{clique problem} and \ref{MESP}.
Because many graph problems are naturally attacked by lifting with edge variables, a natural approach that we are pursuing is a bound
for the maximum-entropy clustering problem based on the \ref{BQP}
bound for \ref{MESP}.

The \emph{maximum-entropy remote sampling problem (MERSP)}
was studied a quarter of a century ago by \cite{AFLW_Remote}. 
Motivated by all of the progress on \ref{MESP} since then,
we are presently developing ideas for MERSP that 
reflect the current algorithmic state-of-the-art for \ref{MESP}.


\section*{Acknowledgments} 
M. Fampa was supported in part by CNPq 
grant\red{s 305444/2019-0 and 434683/2018-3.}\cyan{ 307167/2022-4.}  
J. Lee was supported in part by AFOSR grant FA9550-22-1-0172. 

\bibliographystyle{alpha}
\bibliography{MESP_update}
\end{document}